\documentclass{article}
\usepackage{graphicx} 
\usepackage{amsthm, amsmath, amsfonts,mathtools, fullpage, graphicx, wrapfig, tikz, caption, subcaption}
\usepackage{natbib}
\usepackage{amssymb}
\usepackage{verbatim}
\usepackage{relsize}
\usepackage{listings}
\usepackage{xcolor}
\usepackage{hyperref}
\usepackage{graphicx}
\usepackage{multirow}
\usepackage[linesnumbered,ruled,vlined]{algorithm2e}
\newcommand\independent{\protect\mathpalette{\protect\independenT}{\perp}}
\def\independenT#1#2{\mathrel{\rlap{$#1#2$}\mkern2mu{#1#2}}}
\newtheorem{theorem}{Theorem}

\newtheorem{assumption}{Assumption}
\usepackage{booktabs}
\usepackage{tabularx}

\usepackage{setspace}

 \usepackage[vmargin=1.18in,hmargin=1.1in]{geometry}

\title{Identification and Estimation of Joint Potential Outcome Distributions from a Single Study}
\author{Zach Shahn$^{1,2}$, David Madigan$^{3}$\\
$^1$:CUNY Graduate School of Public Health and Health Policy, New York, NY, United States\\
$^2$:Institute for Implementation Science in Population Health, New York, NY, United States\\
$^2$:Northeastern University, Boston, MA, United States}

\begin{document}
\begin{abstract}
Most causal inference methods focus on estimating marginal average treatment effects, but many important causal estimands depend on the joint distribution of potential outcomes, including the probability of causation and proportions benefiting from or harmed by treatment. \citet{wu2025promises} recently established nonparametric identification of this joint distribution for categorical outcomes under binary treatment by leveraging variation across multiple studies. We demonstrate that their multi-study framework can be implemented within a single study by using a baseline covariate that is associated with untreated potential outcomes but does not modify treatment effects conditional on those outcomes. This reframing substantially broadens the practical applicability of their results, as it eliminates the need for multiple independent datasets and gives analysts control over covariate selection to satisfy key identifying assumptions. We provide complete identification and estimation theory for the single-study setting, including a Neyman-orthogonal estimator for cases where the conditional independence assumption only holds after adjusting for covariates. However, we argue that only in unusual settings would it be even theoretically possible for the identifying assumptions to hold exactly, making sensitivity analysis particularly important. We validate the estimator in a simulation and apply it to data from a large field experiment assessing the effect of mailings on voter turnout.
\end{abstract}
\maketitle

\section{Introduction} Most causal inference methods focus on estimating marginal average treatment effects comparing expected potential outcomes under treatment and no treatment. However, there exist causal estimands of interest that crucially depend on the joint distribution of potential outcomes. Examples include the probability of causation \citep{tian2000probabilities, robins1989probability}, proportion harmed by treatment (probability that treated potential outcome Y(1) = 1 and untreated potential outcome Y(0) = 0, for adverse outcome $Y$), and proportion benefiting from treatment (probability that Y(1) = 0 and Y(0) = 1). The joint distribution of potential outcomes is difficult to estimate (i.e. requires stronger identification assumptions than marginal effects) because both potential outcomes are never observed in the same unit. \citet{wu2025promises} recently established nonparametric identification and inference for the joint distribution of categorical potential outcomes under a binary treatment by leveraging data from multiple studies. Their key assumptions were that the untreated potential outcome Y(0) distribution must vary across studies, and the treated potential outcome Y(1) must be conditionally independent of study given Y(0).

Here, we point out that the assumption of multiple studies is not essential. The same identification strategy can be implemented within a single study by using a discrete baseline covariate S (possibly constructed from other baseline covariates V) that is associated with the untreated potential outcome Y(0), but does not modify treatment effects once Y(0) is conditioned on. In this case, subgroups defined by the levels of S play exactly the same role as the distinct studies in the formulation of \citet{wu2025promises}. This simple reframing greatly broadens the scope of their results: all identification and estimation procedures derived under the multi-study paradigm can be applied without requiring multiple independent datasets.

Moreover, the single-study formulation has distinct advantages. Most importantly, it provides freedom to choose the variable S based on subject-matter knowledge. In the multiple-study setup, the required assumptions must be satisfied by chance due to variation across latent factors that differ between study populations. In practice, there is no guarantee that these differences will satisfy the necessary conditions, and whether they do is largely beyond the analyst's control (though empirically testable with enough studies). By contrast, in a single study, the analyst can deliberately select a baseline covariate that is associated with Y(0) but plausibly not a modifier of treatment effects conditional on Y(0). This flexibility makes the single-study version of the framework both more powerful and more transparent.

In what follows, we develop this idea concretely. In Section \ref{sec:assumptions}, we introduce the setup and assumptions. In Sections \ref{sec:id} and \ref{sec:estimation}, we demonstrate the (obvious) equivalence between the multi-study and single-study setups and show how the identification and estimation results of \citet{wu2025promises} port over to our setting. In Section \ref{sec:covariates}, we provide Neyman-orthogonal \citep{chernozhukov2018double} estimators for when key assumptions only hold conditional on covariates. However, we will argue in Section \ref{sec:causal_pie} that under a minimal sufficient cause or `causal pie' model \citep{rothman1976causes}, the identification assumptions can only ever even theoretically hold under unusual circumstances. Hence, the sensitivity analysis that we discuss in Section \ref{sec:sensitivity} is particularly important. We apply our approach in simulations in Section \ref{sec:simulations} and to a randomized field experiment designed to estimate the effect of mailings on voter turnout \citep{gerber2008social} in Section \ref{sec:application}. The experiment found that mailings on average increased probability of voting \citep{gerber2008social}, but one might further wonder whether the mailings actually caused anyone not to vote. Our methods can help to address that question.

\section{Notation, Setup, and Assumptions}\label{sec:assumptions}

Suppose we observe N independent and identically distributed observations $(A_i, Y_i, V_i)$, $i \in \{1, \ldots, N\}$, with $A_i$ a binary treatment, $Y_i$ a categorical outcome (binary unless otherwise stated), and $V_i$ a baseline covariate vector. Let $Y(a)$ denote the potential outcome \citep{rubin1974estimating} under intervention setting treatment A to value $a$. Interest centers on the causal estimand $\mathbb{P}(Y(0), Y(1))$, i.e. the joint distribution of the potential outcomes. Here, the distribution $\mathbb{P}$ is over the hypothetical superpopulation from which the units are drawn.

We make the usual assumptions required for identification of marginal average treatment effects. First, we make the consistency assumption that the observed outcome Y is equal to the potential outcome Y(A) corresponding to the observed treatment. Initially, to focus on the main ideas, we assume that the data come from a randomized experiment such that the marginal positivity $(0 < \mathbb{P}(A = 1) < 1)$ and exchangeability $((Y(0), Y(1)) \perp\!\!\!\perp A)$ assumptions also hold. In Section \ref{subsec:confounding}, we will consider the observational setting in which the positivity and exchangeability assumptions are made conditional on V.

We will denote by $S = f(V)$ a discrete variable that is some function of the baseline covariates V. S might just be a covariate contained in V that is selected by the analyst, e.g. 'sex'. Alternatively, S might be constructed by the analyst from V, e.g. 'weight tertile' derived from weight in pounds. To identify the joint distribution $\mathbb{P}(Y(0), Y(1))$ of interest, we require the following two assumptions regarding S:

\begin{assumption}\label{assumption:s-rel}
    \textbf{S-relevance:} $S \not\perp\!\!\!\perp Y(0)$
\end{assumption}
\begin{assumption}\label{assumption:s-ci}
    \textbf{S-conditional independence:} $S \perp\!\!\!\perp Y(1)|Y(0)$
\end{assumption}

Assumption \ref{assumption:s-rel} simply states that S is associated with the untreated potential outcomes. Assumption \ref{assumption:s-ci} states that any effect modification by S on any scale is purely via the untreated potential outcome. For example, suppose S is weight tertile, A is an antiglycemic Type 2 Diabetes medication, and Y is occurrence of an adverse cardiovascular event within one year. S would surely satisfy Assumption 1, as cardiovascular events are more common in heavier people. Therefore, unconditional on Y(0), S would also be associated with Y(1) (assuming Y(0) and Y(1) are not independent, as they would surely not be). Furthermore, unconditional on Y(0), S would be an effect modifier on some scale. If the additive treatment effect were constant across weight tertiles, the multiplicative effect would not be, and vice versa. However, depending on the mechanism of action of the drug, it is possible that all effect modification by weight tertile is via the untreated outcome, in which case Assumption \ref{assumption:s-ci} would hold. We further discuss when Assumption \ref{assumption:s-ci} might hold in Section \ref{sec:causal_pie}. We stress that the analyst is free to select or construct S from available baseline covariates as they please so as to maximize the likelihood that these key assumptions hold. It might be that Assumption \ref{assumption:s-ci} only holds conditional on covariates. We will discuss this setting in Section \ref{subsection:het_adj}.

\section{Identification}\label{sec:id}

The main identification result immediately and logically follows from the proof of \citet{wu2025promises} in the context of multiple studies, as each subgroup defined by stratum $S = s$ can just be thought of as a separate 'study'. We repeat their argument here for completeness. For each $s \in \mathcal{S}$, for $\mathcal{S}$ the set of possible values of S, we have:
\begin{align}
\begin{split}\label{eq:id1}
&\mathbb{P}(Y(1) = 1|S = s) = \\
&\mathbb{P}(Y(1) = 1|S = s, Y(0) = 0)\mathbb{P}(Y(0) = 0|S = s) + \mathbb{P}(Y(1) = 1|S = s, Y(0) = 1)\mathbb{P}(Y(0) = 1|S = s).
\end{split}
\end{align}

Assumption \ref{assumption:s-ci} implies that $\mathbb{P}(Y(1) = 1|S = s, Y(0) = 1) = \mathbb{P}(Y(1) = 1|Y(0) = 1)$ and $\mathbb{P}(Y(1) = 1|S = s, Y(0) = 0) = \mathbb{P}(Y(1) = 1|Y(0) = 0)$, and therefore neither quantity depends on s. Thus, (\ref{eq:id1}) can be rewritten:
\begin{equation}\label{eq:id2}
\mathbb{P}(Y(1) = 1|S = s) = \mathbb{P}(Y(1) = 1|Y(0) = 0)\mathbb{P}(Y(0) = 0|S = s) + \mathbb{P}(Y(1) = 1|Y(0) = 1)\mathbb{P}(Y(0) = 1|S = s).
\end{equation}

Note that under the standard consistency, positivity, and exchangeability assumptions, the quantities $\mathbb{P}(Y(1) = 1|S = s)$ and $\mathbb{P}(Y(0) = 0|S = s)$ are each identified as $\mathbb{P}(Y = 1|A = 1, S = s)$ and $\mathbb{P}(Y = 0|A = 0, S = s)$, respectively. (\ref{eq:id2}) is then an equation in two unknowns, namely $\mathbb{P}(Y(1) = 1|Y(0) = 0)$ and $\mathbb{P}(Y(1) = 1|Y(0) = 1)$. As long as $|\mathcal{S}| \geq 2$, i.e. S has at least two levels, we have at least two equations in two unknowns (i.e. (\ref{eq:id2}) at each level of S). This means that $\mathbb{P}(Y(1) = 1|Y(0) = 0)$ and $\mathbb{P}(Y(1) = 1|Y(0) = 1)$ are identified as the unique solution to these equations as long as the following condition holds:

\begin{assumption}\label{assumption:full_rank}
    The matrix $(\mathbb{P}(Y(0) = 0|S = \cdot), \mathbb{P}(Y(0) = 1|S = \cdot))_{|\mathcal{S}| \times 2}$ has full rank.
\end{assumption}

$\mathbb{P}(Y(1), Y(0)|S = s)$ is then of course identified as $\mathbb{P}(Y(0)|S = s)\mathbb{P}(Y(1)|Y(0), S = s)$. And further, $\mathbb{P}(Y(0), Y(1))$ is identified as $\sum_s \mathbb{P}(Y(0)|S = s)\mathbb{P}(Y(1)|Y(0), S = s)\mathbb{P}(S = s)$.

\section{Estimation}\label{sec:estimation}

Following \citet{wu2025promises}, we can estimate $\theta \equiv (\mathbb{P}(Y(1) = 1|Y(0) = 0), \mathbb{P}(Y(1) = 1|Y(0) = 1)) = (\theta_1, \theta_2)$ via least squares as follows. Let $p_s^{(1)} = \mathbb{P}(Y = 1|A = 1, S = s)$, $p_s^{(0)} = \mathbb{P}(Y = 1|A = 0, S = s)$, and $Z_s = (1-p_s^{(0)}, p_s^{(0)})^T$. Then, by (\ref{eq:id2}), $p_s^{(1)} = Z_s^T \theta$. And under the full rank Assumption \ref{assumption:full_rank},
\begin{equation}
\theta = \left(\frac{1}{|\mathcal{S}|} \sum_{s=1}^{|\mathcal{S}|} Z_s Z_s^T\right)^{-1} \cdot \frac{1}{|\mathcal{S}|} \sum_{s=1}^{|\mathcal{S}|} Z_s p_s^{(1)}.
\end{equation}

Thus, the plug-in estimator
\begin{equation}
\hat{\theta} = \left(\frac{1}{|\mathcal{S}|} \sum_{s=1}^{|\mathcal{S}|} \hat{Z}_s \hat{Z}_s^T\right)^{-1} \cdot \frac{1}{|\mathcal{S}|} \sum_{s=1}^{|\mathcal{S}|} \hat{Z}_s \hat{p}_s^{(1)},
\end{equation}
in which $p_s^{(1)}$ and $Z_s$ have been replaced by unbiased estimators (e.g. sample proportions), is consistent. Furthermore, $\hat{\theta}$ is asymptotically normal, satisfying
\begin{equation*}
\sqrt{n}(\hat{\theta} - \theta) \xrightarrow{d} N(0, \sigma^2),
\end{equation*}
where $\sigma^2 = C^{-1}DC^{-1}$,
\begin{equation*}
C = \frac{1}{|\mathcal{S}|} \sum_{s=1}^{|\mathcal{S}|} Z_s Z_s^T,
\end{equation*}
and
\begin{align*}
D &= \frac{1}{|\mathcal{S}|^2} \sum_{s=1}^{|\mathcal{S}|} \text{Var}\left(\frac{\mathbb{I}(Y = 0, S = s, A = 0) - \mathbb{P}(Y = 0, S = s, A = 0)}{\mathbb{P}(S = s, A = 0)}\right. \\
&\quad + \frac{\mathbb{I}(Y = 1, S = s, A = 0) - \mathbb{P}(Y = 1, S = s, A = 0)}{\mathbb{P}(S = s, A = 0)} \\
&\quad \left. + \frac{\mathbb{I}(Y = 1, S = s, A = 1) - \mathbb{P}(Y = 1, S = s, A = 1)}{\mathbb{P}(S = s, A = 1)} Z_s\right).
\end{align*}\\
\\
Standard errors can also be estimated via nonparametric bootstrap.

\section{Introduction of Covariates}\label{sec:covariates}

There are two main ways covariates might enter into the analysis beyond the selection and definition of S. First, it might be necessary to adjust for covariates to satisfy the exchangeability assumption. Second, it could be that the conditional independence Assumption \ref{assumption:s-ci} only holds conditional on covariates. Confounding adjustment is of course much more likely to be of concern in observational studies, while conditioning to satisfy Assumption \ref{assumption:s-ci} could easily be required in any setting.

\subsection{Confounding Adjustment}\label{subsec:confounding}

Suppose the exchangeability assumption required to identify ATEs only holds conditional on baseline covariates $V_C \subset V$, i.e.
\begin{equation}
Y(a) \perp\!\!\!\perp A|V_C.
\end{equation}

Here, $V_C$ is a vector of confounding adjustment covariates. Further, suppose for simplicity (for now) that Assumption \ref{assumption:s-ci} holds unconditional on any covariates. Then (\ref{eq:id2}), which did not depend on exchangeability, still holds. The only difference from the randomized experiment setting is that $\mathbb{P}(Y(1)|S = s)$ and $\mathbb{P}(Y(0) = 1|S = s)$ are now identified by the g-formula
\begin{equation}\label{eq:gform}
\mathbb{P}(Y(a) = 1|S = s) = \int_v \mathbb{P}(Y = 1|A = a, V_C = v, S = s)p_{V_C|S}(v|S = s)dv
\end{equation}
instead of simple sample averages. The functionals (\ref{eq:gform}) can be estimated by standard confounding adjustment techniques, such as standardization, inverse probability of treatment weighting, or DML (doubly robust estimation with nonparametric machine learning estimation of nuisance functions, sample splitting, and crossfitting). Then the same plug-in estimator from Section \ref{sec:estimation} can be computed with covariate adjusted estimates of relevant counterfactual risks in place of sample proportions. Standard errors for all quantities of interest can be estimated via nonparametric bootstrap, but we give asymptotic results below.

\subsubsection*{Asymptotic distribution of $\hat{\theta}$ with confounding adjustment}

We now establish the asymptotic distribution of $\hat{\theta}$ when the stratum-specific risks $p_s^{(a)} \equiv \mathbb{P}(Y(a) = 1|S = s)$ are estimated using cross-fitted doubly-robust (DML) estimators. We will derive a sandwich representation of the variance that is conceptually very similar to the result of \citet{wu2025promises} in the setting with no confounding. The only difference is that we substitute influence functions of confounding adjusted estimators of the counterfactual risks in place of the influence functions of sample proportions that lead to the result of Wu and Mao (2025) in the no confounding setting. Let $\theta_i$ denote the $i$th component of $\theta$, so that $\theta_1$ is $\mathbb{P}(Y(1) = 1|Y(0) = 0)$ and $\theta_2$ is $\mathbb{P}(Y(1) = 1|Y(0) = 1)$. Let $\mu_0 = \mathbb{P}(Y(0) = 0)$ and $\mu_1 = \mathbb{P}(Y(0) = 1)$. Let $\phi_{a,s}$ denote the influence function of $\hat{p}_s^{(a)}$ for $a \in \{0,1\}$, and let $\phi_s = (\phi_{0,s}, \phi_{1,s})^T$.

\begin{theorem}
    
 Under regularity conditions,
\begin{equation*}
\sqrt{n}(\hat{\theta} - \theta) \xrightarrow{d} N(0, C^{-1}\Sigma C^{-1}),
\end{equation*}
where
\begin{equation*}
\Sigma = \text{Var}\{\psi(O)\}, \quad \psi(O) = Z_S\left[U_1(O) - p_S^{(1)} - (\theta_2 - \theta_1)\{U_0(O) - p_S^{(0)}\}\right],
\end{equation*}
with
\begin{equation*}
U_a(O) = \mu_a(V, S) + \frac{\mathbb{I}\{A = a\}}{\pi_a(V, S)}(Y - \mu_a(V, S)), \quad a \in \{0,1\}
\end{equation*}
and $C = -\frac{1}{|\mathcal{S}|}\sum_s Z_s Z_s^T$.
\end{theorem}
\begin{proof}
    See Appendix
\end{proof}
\subsubsection*{Delta method for derived quantities}

Often $\theta$ is not of immediate interest, but rather a means to compute other features of the joint distribution of $(Y(0), Y(1))$ for which we also desire standard errors. Define the joint distribution of $(Y(0), Y(1))$ as
\begin{align*}
\mathbb{P}(Y(0) = 0, Y(1) = 0) &= (1 - \theta_1)\mu_0, \\
\mathbb{P}(Y(0) = 0, Y(1) = 1) &= \theta_1\mu_0, \\
\mathbb{P}(Y(0) = 1, Y(1) = 0) &= (1 - \theta_2)\mu_1, \\
\mathbb{P}(Y(0) = 1, Y(1) = 1) &= \theta_2\mu_1.
\end{align*}

Let $\eta = (\theta_1, \theta_2, \mu_0, \mu_1)^T$. Each of the four joint probabilities is a smooth function $T_j(\eta)$. By the Delta method,
\begin{equation*}
\sqrt{n}\{T_j(\hat{\eta}) - T_j(\eta)\} \xrightarrow{d} N(0, \nabla T_j(\eta)^T \Omega \nabla T_j(\eta)),
\end{equation*}
where $\Omega$ is the covariance matrix of the joint influence function of $(\hat{\theta}_1, \hat{\theta}_2, \hat{\mu}_0, \hat{\mu}_1)$. Explicit gradients are:
\begin{align*}
\nabla T_{00}(\eta) &= (-\mu_0, 0, 1 - \theta_1, 0)^T, \\
\nabla T_{01}(\eta) &= (\mu_0, 0, \theta_1, 0)^T, \\
\nabla T_{10}(\eta) &= (0, -\mu_1, 0, 1 - \theta_2)^T, \\
\nabla T_{11}(\eta) &= (0, \mu_1, 0, \theta_2)^T.
\end{align*}

Thus each entry of the joint distribution has a valid influence-function variance estimator. Covariances between entries can be obtained similarly by taking products of gradients with $\Omega$. However, we again stress that nonparametric bootstrap would be a simple and practical option for estimating standard errors of these quantities.

\subsection{`Effect Heterogeneity Adjustment'}\label{subsection:het_adj}

Suppose that Assumption \ref{assumption:s-ci} only holds conditional on a vector of covariates $V_S \subset V$, i.e.

\begin{assumption}\label{assumption:v-s-ci}
    $Y(1) \perp\!\!\!\perp S|Y(0), V_S$.
\end{assumption} 

$V_S$ would typically comprise effect modifiers (that remain effect modifiers even conditional on $Y(0)$) that are associated with S. For example, suppose the mechanism of action of a treatment is different in men and women. Then, without conditioning on any covariates, selecting S to be 'weight tertile' would be a poor choice, since Assumption \ref{assumption:s-ci} would fail due to the association between weight and sex. However, conditional on $V_S = \{\text{sex}\}$, Assumption \ref{assumption:v-s-ci} might hold. As noted by \citet{wu2025promises}, under Assumption \ref{assumption:v-s-ci}, the identification argument from Section \ref{sec:id} goes through within levels of $V_S$. That is,
\begin{align*}
\begin{split}\label{eq:covariate:id}
&\mathbb{P}(Y(1) = 1|S = s, V_S = v) \\
&\quad = \mathbb{P}(Y(1) = 1|Y(0) = 0, V_S = v)\mathbb{P}(Y(0) = 0|S = s, V_S = v) \\
&\qquad + \mathbb{P}(Y(1) = 1|Y(0) = 1, V_S = v)\mathbb{P}(Y(0) = 1|S = s, V_S = v)
\end{split}
\end{align*}
for all $v \in \mathcal{V}_S$. If $V_S$ is low dimensional and categorical, estimation might also proceed as in Section \ref{sec:estimation} within levels of $V_S$. That is, we can estimate $\mathbb{P}(Y(1), Y(0)|S = s, V_S = v)$ for all joint levels $(s, v)$ of $(S, V_S)$ as in Section \ref{sec:estimation} and then obtain estimates of $\mathbb{P}(Y(1), Y(0)|S = s)$ and $\mathbb{P}(Y(1), Y(0))$ via standardization. Standard errors could be estimated by nonparametric bootstrap, or the analytic variance estimates from Section \ref{sec:estimation} could be computed for each stratum of $V_S$ and combined via the Delta method.

However, if $V_S$ is high dimensional or contains a continuous covariate, the curse of dimensionality would require parametric assumptions. We assume that
\begin{align}
\mathbb{P}(Y(1) = 1|Y(0) = 0, V_S = v) &= g(v; \beta), \\
\mathbb{P}(Y(1) = 1|Y(0) = 1, V_S = v) &= h(v; \lambda)
\end{align}
for finite dimensional parameters $\beta$ and $\lambda$. For example, $g(v; \beta)$ and $h(v; \lambda)$ could denote linear functions $\beta^T b(v)$ and $\lambda^T b(v)$, respectively in known basis $b(v)$ (e.g. $b(v) = (1, v)$).

One approach to estimation is as follows. For a range of $s \in \mathcal{S}$ and $v \in \mathcal{V}_S$, compute estimates $\hat{p}_{s}^{(a)}(v)$ of $p_{s}^{(a)}(v) \equiv \mathbb{P}(Y(a) = 1|S = s, V_S = v)$ using your favorite regression. Then, solve the least squares problem
\begin{equation*}
(\hat{\beta}, \hat{\lambda}) = \arg\min_{\beta,\lambda} \sum_{s,v} \{\hat{p}_{s}^{(1)}(v) - g(v; \beta)(1 - \hat{p}_{s}^{(0)}(v)) - h(v; \lambda)\hat{p}_{s}^{(0)}(v)\}^2.
\end{equation*}

Let $\xi = (\beta, \lambda)$ denote the parameters of the conditional model. With estimates $\hat{\xi} = (\hat{\beta}, \hat{\lambda})$ in hand, we can consistently estimate $\theta_1 = \mathbb{P}(Y(1) = 1|Y(0) = 0)$ and $\theta_2 = \mathbb{P}(Y(1) = 1|Y(0) = 1)$ by standardization as
\begin{align*}
\theta_1 &= \mathbb{P}(Y(1) = 1|Y(0) = 0) = \int_v g(v; \hat{\beta})p(v)dv, \\
\theta_2 &= \mathbb{P}(Y(1) = 1|Y(0) = 1) = \int_v h(v; \hat{\lambda})p(v)dv.
\end{align*}

Standard errors can be estimated via bootstrap if none of the regressions to generate $\hat{p}_{s}^{(a)}$ are data adaptive. However, if $\hat{p}_{s}^{(a)}$ comes from a data adaptive machine learning approach, we will need a Neyman-orthogonal estimator of $\xi = (\beta, \lambda)$ to generate valid confidence intervals.

\subsubsection*{A Neyman Orthogonal Estimator}

Let $p^{(0)} = \{p_{s}^{(0)} : s \in \mathcal{S}\}$ and $p^{(1)} = \{p_{s}^{(1)} : s \in \mathcal{S}\}$. Let $\eta = (p^{(0)}, p^{(1)}, \pi_a(s, v) = \mathbb{P}(A = a|S = s, V = v))$ denote the nuisance functions. 

\begin{theorem}\label{theorem2}
Suppose Assumption~\ref{assumption:v-s-ci} holds and that
$g(v;\beta)$ and $h(v;\lambda)$ are linear in a known basis $b(v)$.
Define
\[
Z_s(v; p_s^{(0)}) \;=\;
\begin{pmatrix}
(1 - p_s^{(0)}(v))\, b(v) \\
p_s^{(0)}(v)\, b(v)
\end{pmatrix}
\in \mathbb{R}^{2p}.
\]
Then the efficient influence function (EIF) for
$\xi = (\beta,\lambda) \in \mathbb{R}^{2p}$ is
\[
\mathrm{EIF}_\xi(O) \;=\; -\,C^{-1}\,\psi(O;\xi,\eta),
\]
where
\begin{align*}
\psi(O;\xi,\eta)
&= Z_S(V; p_S^{(0)}) \Bigg\{
  p_S^{(1)}(V) +
  \frac{\mathbf{1}\{A=1\}}{\pi_1(S,V)}\big(Y - p_S^{(1)}(V)\big)
  - Z_S(V;p_S^{(0)})^\top \xi
\Bigg\} \\
&\quad + Z_S(V;p_S^{(0)})\, b(V)^\top (\beta - \lambda)\,
\frac{\mathbf{1}\{A=0\}}{\pi_0(S,V)}\big(Y - p_S^{(0)}(V)\big),
\end{align*}
and
\[
C \;=\; E\!\left[-\,Z_S(V;p_S^{(0)}) Z_S(V;p_S^{(0)})^\top\right].
\]

The estimator $\hat\xi = (\hat\beta,\hat\lambda)$ defined by the
estimating equation
\[
P_n \psi(O;\hat\xi,\hat\eta) = 0
\]
is asymptotically linear with influence function $\mathrm{EIF}_\xi$ and satisfies
\[
\sqrt{n}(\hat\xi - \xi)
\;\overset{d}{\to}\;
N\!\left(0,\; C^{-1}\, Var\!\big(\psi(O;\xi,\eta)\big)\, C^{-1}\right).
\]
Moreover, the score $\psi(O;\xi,\eta)$ is Neyman-orthogonal in the nuisance
functions $\eta$.
\end{theorem}

\begin{proof}
    See Appendix.
\end{proof} 

By results in \citet{chernozhukov2018double}, it follows that if $\hat{\xi} = (\hat{\beta}, \hat{\lambda})$ is estimated by a sample splitting and cross-fitting procedure, nonparametric bootstrap would produce valid confidence intervals for $\beta$, $\lambda$, and derived quantities such as the probabilities comprising the joint distribution $(Y(0), Y(1))$. Alternatively, empirical estimates of the asymptotic variance approximation from Theorem \ref{theorem2} could be used to construct confidence intervals for $\beta$ and $\lambda$, which could be extended to derived quantities via the Delta method.

Theorem \ref{theorem2} assumes a linear probability model. In the Appendix, we extend Theorem \ref{theorem2} to generalized linear models, where the inverse logit link would be of primary interest. In the generalized linear model case, the estimating function $\psi(O;\xi,\eta)$ retains the same structure and is symbolically identical to the expression in Theorem \ref{theorem2}. However, the definition of the design vector $Z_s(v; p_s^{(0)})$ changes so that it is a function of the parameter of interest $\xi$, which necessitates an iterative approach to estimation.

\section{Assessment of Assumption Plausibility}\label{sec:causal_pie}
We will argue that under a minimum sufficient cause or `causal pie' model \citep{rothman1976causes}, it is usually impossible for Assumptions \ref{assumption:s-rel} and \ref{assumption:v-s-ci} to jointly hold exactly. Under a causal pie model, the outcome occurs whenever any minimal sufficient set (or `pie') of causal conditions (slices) is satisfied. The effect of treatment on the outcome is realized by activating or de-activating some subset of causal pie slices, which in turn determine whether a causal pie is satisfied. Individual treatment effects depend in large part on which slices would be satisfied absent treatment. As a rough illustrative example, the outcome `skin cancer' might have one causal pie comprising: {`high UV exposure', `fair skin', `genetic proclivity'}. Sunscreen can prevent skin cancer in some fair skinned people by de-activating the `high UV exposure' slice, but would have no effect on those lacking genetic proclivity. 

Let $Z$ denote the set of events comprising all the slices of any causal pie for $Y$. Let $Z_S\subseteq Z$ denote the set of events such that $Z_S(0) \not\perp\!\!\!\perp S|V_S$ where $Z_S(0)$ denotes the untreated counterfactual values of $Z_S$. One way for Assumption \ref{assumption:v-s-ci} to hold would be if treatment broke the associations between $S$ and $Z$. Another way to put this is that $S$ is d-connected to some elements of $Z$ given $V_S$ in the $A=0$ Single World Intervention Graph (SWIG) \citep{richardson2013single} but not in the $A=1$ SWIG. For example, perhaps $Y$ is `attends a scheduled medical visit', $Z_S$ is `financial burden of transportation to a healthcare facility', $S$ is `income', and $A$ is `distribution of a free transit pass'. When $A=1$, $Z_S(1)$ is deterministically 0 (i.e. `no burden') and thus unassociated with $S$. 

While it is possible for situations such as described above to occur, we submit that it is far more common that $S$ is conditionally associated with $Z(1)$ given $V_S$ if and only if $S$ is conditionally associated with $Z(0)$ given $V_S$, and we assume as much going forward. Now, we will consider two cases. First, suppose that $Z_S(1)\neq Z_S(0)$, i.e. $A$ affects some variable(s) $Z_{SA}\in Z_S$. Because $S$ is conditionally associated with $Z_{SA}(1)$ by assumption, it is also conditionally associated with $Y(1)$ via the path $S\leftarrow\rightarrow Z_{SA}(1)\rightarrow Y(1)$, implying that Assumption \ref{assumption:v-s-ci} is violated. As the second case, suppose that no variable in $Z_S$ is affected by $A$, i.e. $Z_S(1)=Z_S(0)$. If $A$ has any effect on $Y$, then there is at least one variable $Z_A\in Z\setminus Z_S$ such that $Z_A(1)\neq Z_A(0)$. The value of $Z_A(1)$ is not a deterministic function of $V_S$ and $Y(0)$ when $A=0$. Assume that, conditional on $V_S$, the potential outcomes $Z_A(1)$ and $Z_A(0)$ are still associated. Then, because $Z_A(0)$ causes $Y(0)$ and Assumption \ref{assumption:v-s-ci} conditions on $Y(0)$, the collider path $S\leftarrow\rightarrow Z_S\rightarrow Y(0)\leftarrow Z_A(0)\leftarrow\rightarrow Z_A(1) \rightarrow Y(1)$ is open, again violating Assumption \ref{assumption:v-s-ci}.  

The above demonstrates that except for unusual circumstances, Assumptions \ref{assumption:s-rel} and \ref{assumption:v-s-ci} cannot both hold under a causal pie model. However, violations, especially those conveyed through collider paths, can be quite small. They can be smaller still if $V_S$ is a good proxy for $Z(1)$. Moreover, adjusting for a rich set $V_S$ might make potential outcome of $Z_A$ closer to conditionally independent. Given this state of affairs, sensitivity analysis is essential.

\section{Sensitivity Analysis}\label{sec:sensitivity}
 Based on our experience, we believe that most choices of $S$ in most applications would only weakly violate Assumption \ref{assumption:v-s-ci}. Therefore, if the estimator is not very sensitive to small violations, it might enjoy wide applicability. We propose to follow the general approach to sensitivity analysis described in \citep{robins2000sensitivity}. 
 
 Absent covariates, define sensitivity parameters $\{\gamma_{sy}:s\in\mathcal{S},y\in\{0,1\}\}$ where $\gamma_{sy}$ represents the difference $P(Y(1)=1|Y(0)=y,S=s)-P(Y(1)=1|Y(0)=y)$, i.e. the magnitude of violation of Assumption \ref{assumption:s-ci}. With these sensitivity parameters specified and therefore known, we can substitute $P(Y(1)=1|Y(0)=y)+\gamma_{sy}$ in place of $P(Y(1)=1|Y(0)=y)$ in identification formula (\ref{eq:id1}):
 \begin{align}
\begin{split}\label{eq:sens_id}
&\mathbb{P}(Y(1) = 1|S = s) = \\
&(\mathbb{P}(Y(1)=1|Y(0)=0)+\gamma_{s0})\mathbb{P}(Y(0) = 0|S = s) + (\mathbb{P}(Y(1)=1|Y(0)=y)+\gamma_{s1})\mathbb{P}(Y(0) = 1|S = s).
\end{split}
\end{align}
For known $\gamma_{sy}$, (\ref{eq:sens_id}) represents $|\mathcal{S}|$ equations in two unknowns, and $\mathbb{P}(Y(1)=1|Y(0)=y)$ for $y\in\{0,1\}$ is identified as the unique solution to these equations under Assumption \ref{assumption:full_rank}. Estimation can proceed by least squares as described in Section \ref{sec:estimation}. Researchers can then assess how conclusions change as $\gamma_{sy}$ varies over a plausible range of values.

With covariates, one can instead specify sensitivity functions $\{\gamma_{sy}(v):s\in\mathcal{S},y\in\{0,1\}\}$ where $\gamma_{sy}(v) =P(Y(1)=1|Y(0)=y,V_S=v,S=s)-P(Y(1)=1|Y(0)=y,V_S=v)$, i.e. the magnitude of violation of Assumption \ref{assumption:v-s-ci} at each level of $V_S$. Again, for any given specification of the sensitivity functions assumed known, the formula (\ref{eq:covariate:id}) with $P(Y(1)=1|Y(0)=y,V_S=v)+\gamma_{sy}(v)$ in place of $P(Y(1)=1|Y(0)=y,V_S=v)$ would still lead to identification of $P(Y(1)=1|Y(0)=y)$, and estimators could be modified accordingly. And again, researchers could assess how conclusions change as $\gamma_{sy}(v)$ varies over a plausible range.

\section{Simulation}\label{sec:simulations}
\subsection{Data Generating Process}\label{subsection:dgp}

We simulated data under a randomized binary treatment $A \in \{0,1\}$, a binary outcome $Y$, and baseline covariates $(S, V_S, X_{\text{other}})$. The process was designed so that effect heterogeneity exists through $V_S$, and conditional independence of the counterfactuals only holds after conditioning on $V_S$. That is, Assumption \ref{assumption:v-s-ci} holds, but not Assumption \ref{assumption:s-ci}.  

\paragraph{Baseline covariates.} 
The effect-modifying covariate $V_S$ was drawn from a truncated normal distribution:
\[
V_S \sim \text{Normal}(0,1), \quad V_S \in [-2,2].
\]
An auxiliary baseline covariate was generated as
\[
X_{\text{other}} \sim \text{Normal}(.25\times V_S, 1),
\]
so that it is correlated with $V_S$. The stratification variable $S$ was constructed by dividing $X_{\text{other}}$ into quartiles and assigning values $S \in \{1,2,3,4\}$.

\paragraph{Treatment assignment.}
Treatment was randomized with probability one-half, independent of the covariates:
\[
A \sim \text{Bernoulli}(0.5).
\]

\paragraph{Potential outcomes.}
The untreated potential outcome was generated from a logistic regression depending on $S$ and $V_S$:
\[
\text{logit}\,\Pr(Y(0)=1 \mid S, V_S) = -0.5 + 0.3(S-2) + 0.2 V_S.
\]
The treated potential outcome was generated according to a linear probability structural model:
\[
\Pr(Y(1)=1 \mid Y(0)=0, V_S) = \beta_0 + \beta_1 V_S,
\]
\[
\Pr(Y(1)=1 \mid Y(0)=1, V_S) = \lambda_0 + \lambda_1 V_S,
\]
with true parameters $\beta=(0.3,\,0.1)$ and $\lambda=(0.7,\,-0.05)$. These choices yield probabilities in the range $[0.1,0.85]$ across the support of $V_S$. We note that, consistent with the discussion in Section \ref{sec:causal_pie}, the data generating process does not follow a causal pie model.

\paragraph{Observed outcome.}
The observed outcome was
\[
Y = (1-A) Y(0) + A Y(1).
\]
\subsection{Estimation}
In this setting, the treatment probabilities were known to be $0.5$ and so were not estimated. Our model for $Y(0)$ was a logistic regression with $S$, a third degree spline basis of $V_S$, and full interactions as predictors. While this model contains the true DGP for $Y(0)$, the true DGP of $Y(1)$ is quite complex and not exactly contained in our nuisance model. However, we note that our model must have approximated the true DGP well, since bias approached 0 as sample size increased in the simulations. We specified correct linear models for $g(V;\beta)$ and $h(V;\lambda)$. We then solved the estimation equations from Theorem \ref{theorem2} to obtain estimates of parameters $\beta$, $\lambda$, $P(Y(1)=1|Y(0)=0)$, and $P(Y(1)=1|Y(0)=1)$.

\subsection{Results}
We assessed bias and coverage of 95\% bootstrap confidence intervals for: $\beta$, $\lambda$, and the derived quantities $P(Y(1)=1|Y(0)=0)$ and $P(Y(1)=1|Y(0)=1)$. The results are displayed in Table \ref{tab:sim-results}. We see that bias is small and coverage is near nominal. Figure \ref{fig:sim-results} and the Empirical SE column of Table \ref{tab:sim-results} show that there was considerable uncertainty about the parameters of interest, even with n=10,000 observations. An interesting question for future research is to what degree these results are typical. Perhaps significantly more data in addition to stronger assumptions are required to obtain precise estimates of features of the joint distribution of potential outcomes compared to marginal expectations. Increasing the strength of association between $V_S$ and $S$ increased sampling variation even further.  

\begin{table}[htbp]
\centering
\caption{Simulation results for the orthogonal estimator under the data generating process described in Section~X. Results are based on $n=10{,}000$ and 500 Monte Carlo replications.}
\label{tab:sim-results}
\begin{tabular}{lrrrrr}
\toprule
Parameter & True & Mean Estimate & Bias & Empirical SE & Coverage (95\% CI) \\
\midrule
$\beta_0$ & 0.300 & 0.290 & -0.010 & 0.051 & 0.950 \\
$\beta_1$ & 0.100 & 0.100 & -0.000 & 0.042 & 0.948 \\
$\lambda_0$ & 0.700 & 0.714 & 0.014 & 0.074 & 0.956 \\
$\lambda_1$ & -0.050 & -0.053 & -0.003 & 0.061 & 0.940 \\
$\Pr(Y(1)=1 \mid Y(0)=0)$ & 0.300 & 0.291 & -0.009 & 0.050 & 0.952 \\
$\Pr(Y(1)=1 \mid Y(0)=1)$ & 0.700 & 0.712 & 0.012 & 0.069 & 0.956 \\
\bottomrule
\end{tabular}
\end{table}

\begin{figure}
    \centering
    \includegraphics[scale=0.6]{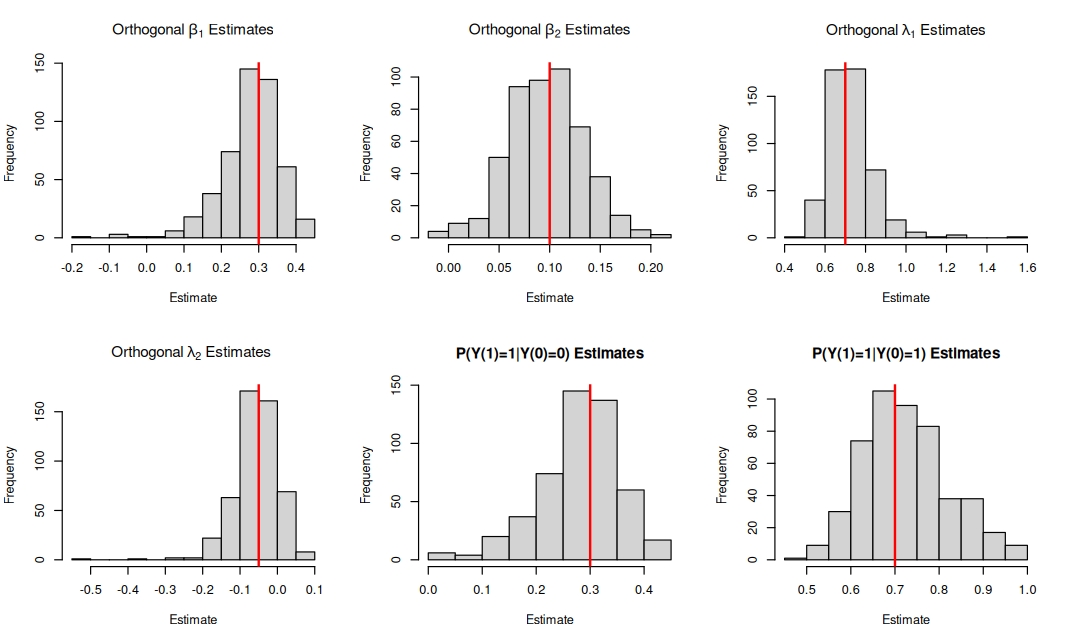}
    \caption{Histograms of parameter estimates across 500 simulations with true values in red}
    \label{fig:sim-results}
\end{figure}

\section{Application to Get Out The Vote Experiment}\label{sec:application}

\subsection*{Data and Setup}
We applied our methods to a large randomized field experiment that mailed voter‐mobilization appeals (``GOTV mailings'') to households, with no mailing as control \citep{gerber2008social}. The unit of analysis is the individual; the treatment $A\in\{0,1\}$ indicates assignment to \emph{any} mailing versus control, and the outcome $Y\in\{0,1\}$ is an indicator of voter turnout in the focal election (a primary in 2006). The dataset contained pre‐treatment covariates: age, sex, prior voting history in several prior elections, number of household members, and voting rates in the participant's geographic region.

Our targets are the conditional probabilities
\[
\Pr\{Y(1)=1\mid Y(0)=0\}\quad\text{and}\quad \Pr\{Y(1)=1\mid Y(0)=1\},
\]
which we interpret, respectively, as an ``activation'' probability among nonvoters under control and a ``no‐harm'' probability among voters under control.

\subsection*{Initial Least–Squares Screen and Motivation}
We first implemented the unadjusted least–squares (LS) estimator from Section \ref{sec:estimation}, trying out multiple different binary choices for $S$: sex (male/female), a binary age indicator (above/below the sample median), and indicators for voting in prior elections (yes/no). Across all of these $S$ choices, the LS estimator yielded $\widehat{\Pr}\{Y(1)=1\mid Y(0)=1\}>1$, i.e., estimates outside $[0,1]$. This pattern is consistent with a failure of Assumption \ref{assumption:s-ci} (conditional independence without further adjustment), motivating use of our orthogonal, covariate‐adjusted procedure from Section \ref{subsection:het_adj}.

\subsection*{Estimation Adjusting for Covariates}
For the orthogonal adjusted analysis, we specified $S$ to be a four‐level factor crossing sex and age (male/old, male/young, female/old, female/young). We chose $|\mathcal S|=4$ to match the four‐dimensional parameter $\xi=(\beta^\top,\lambda^\top)^\top$, defined below. We constructed $V_S$ as a \emph{prognostic score} for $Y(0)$, defined as $\widehat{\Pr}\{Y(0)=1\mid V\setminus\{\text{age, sex}\}\}$. We estimated this prognostic score via logistic regression applied to the control group and using prior voting indicators, the number of household members, and average past voter participation in the geographic area as predictors. We specified $g(v;\beta)=\operatorname{expit}(\beta^\top b(v))$ and $h(v;\lambda)=\operatorname{expit}(\lambda^\top b(v))$ with $b(v)=(1,v)^\top$ (linear index with logistic link). We estimated $\pi$ by the sample proportion assigned to mailers (consistent with randomization). The nuisance functions
  \[
    p_a(s,v)=\Pr(Y=1\mid A=a,S=s,V_S=v),\qquad a\in\{0,1\},
  \]
  were fit via logistic regression including $S$, a spline expansion of $V_S$, and their interaction as predictors. We estimated $\xi=(\beta,\lambda)$ by solving the logistic orthogonal estimating equations of Theorem 
  \ref{theorem3}. Standard errors were estimated using bootstrap clustered at the household level.

\subsection*{Results}
The orthogonal adjusted logistic estimator yielded
\[
\widehat{\Pr}\{Y(1)=1\mid Y(0)=1\} = 1.00,
\qquad
\widehat{\Pr}\{Y(1)=1\mid Y(0)=0\} = 0.054.
\]
Unfortunately, the point estimate $\widehat{\Pr}\{Y(1)=1\mid Y(0)=1\}=1$ near the boundary of the parameter space invalidated our nice normal asymptotic theory and bootstrap confidence intervals. To still obtain valid confidence intervals, an analyst might resort to bootstrap variants such as $BC_a$ \citep{diciccio1996bootstrap}. Alternatively, one might take $\widehat{\Pr}\{Y(1)=1\mid Y(0)=1\}=1$ as known and obtain a corresponding confidence interval for $\widehat{\Pr}\{Y(1)=1\mid Y(0)=0\}$. 

Here, for illustrative purposes, we simply discuss interpretation of the point estimates. The estimate $\widehat{\Pr}\{Y(1)=1\mid Y(0)=1\}=1$, in words, means that no individual who would have voted under control was deterred by receiving a mailing (i.e., no estimated harm among always‐voters). A political operative might conclude that if mailings are cheap, they might as well be sent to everyone as there would be no benefit from targeting. The activation probability point estimate $\widehat{\Pr}\{Y(1)=1\mid Y(0)=0\} = 0.054$ says that among individuals who would not have voted under control, the mailings increased turnout by about 5.4 percentage points.  

As an internal consistency check, note that the law of total probability implies
\[
\Pr\{Y(1)=1\}
=
\Pr\{Y(0)=0\}\cdot \Pr\{Y(1)=1\mid Y(0)=0\}
+
\Pr\{Y(0)=1\}\cdot \Pr\{Y(1)=1\mid Y(0)=1\}.
\]
Using the trial estimate of $\Pr\{Y(1)=1\}$ together with our conditional estimates, the implied value of $\Pr\{Y(0)=1\}$ closely matched the trial’s direct estimate of $\Pr\{Y(0)=1\}$, supporting the plausibility of the conditional independence Assumption \ref{assumption:v-s-ci} after adjusting for the prognostic score $V_S$.

\section{Discussion and Future Work}

In this paper, we have transported the results of \citet{wu2025promises} to a single study setting. In this setting, the investigator has the flexibility to choose a baseline covariate on which to define subgroups such that the key conditional independence assumption is likely to hold. Furthermore, if the conditional independence assumption is more likely to approximately hold conditional on other covariates, we have provided a robust procedure to estimate quantities of interest and their confidence intervals under minimal assumptions on nuisance regressions. This is at the price of specifying parametric models for the dependence of $P(Y(1)=1|Y(0)=a,V_S)$ on adjustment covariate $V_S$. However, we have established that there are theoretical reasons to doubt that the identification assumptions hold exactly (to go along with the usual common sense reasons that are also present in other contexts), making sensitivity analysis particularly important.

There are still many interesting open questions. Perhaps most importantly, the selection of S (including which variable to use and how to discretize it to define `sub-studies') introduces methodological challenges. Adaptive or ensemble-based procedures may help to identify promising $S$ candidates, but raise new issues for valid inference. 

In specifying the dimension and level of flexibility of $g(v;\beta)$ and $h(v;\beta)$, there will be a tradeoff between high variance and model mis-specification error for high dimensional and/or flexible choices and bias from violation of Assumption \ref{assumption:v-s-ci} for low dimensional or inflexible choices. A promising strategy could be to use a spline expansion of a prognostic score $\hat{P}(Y(0)=1|V/S)$ for $V_S$. This was the approach we took in our real data application in Section \ref{sec:application}. The entire issue of model selection requires more investigation and ties in with the selection of the number of levels of $S$.

Fitting the estimators into a Bayesian framework would also be desirable. A Bayesian approach could further facilitate bias analysis and enable quantities such as probability of causation to enter into formal decision analyses.


\bibliography{bibliography}

\section*{Appendix}
\subsection*{Proof of Theorem 1}
\begin{proof} $\theta$ solves $m(\theta) = \sum_s g_s(\theta) \equiv \sum_s Z_s(p_s^{(1)} - Z_s^T\theta) = 0$. By Taylor expansion, $g_s(\hat{\theta}) = g_s(\theta) + \frac{\partial g_s(\theta)}{\partial(p_s^{(1)}, p_s^{(0)})} + o_p(n^{-1/2})$. Now,
\begin{equation*}
\frac{\partial g_s(\theta)}{\partial(p_s^{(1)}, p_s^{(0)})} = Z_s[1, \theta_1 - \theta_2].
\end{equation*}
Hence, the unit level contribution to stratum s is
\begin{equation*}
\psi_s(O) = \frac{\partial g_s(\theta)}{\partial(p_s^{(1)}, p_s^{(0)})} \phi_s(O) = Z_s[(\theta_2 - \theta_1)\phi_{0,s} + \phi_{1,s}]
\end{equation*}
where we have substituted in the influence functions $\phi_s$ for DML AIPW estimators of the counterfactual risks. M-estimation theory then gives
\begin{equation*}
\sqrt{n}(\hat{\theta} - \theta) = C^{-1} \frac{1}{\sqrt{n}} \sum_{i=1}^n \psi(O_i) + o_p(1)
\end{equation*}
for $C = \mathbb{E}[\frac{1}{|\mathcal{S}|}\sum_s \frac{\partial\{Z_s(p_s^{(1)} - Z_s^T\theta)\}}{\partial\theta}] = -\frac{1}{|\mathcal{S}|}\sum_s Z_s Z_s^T$, and the CLT completes the proof.
\end{proof}
\subsection*{Proof of Theorem 2}

\subsubsection*{EIF Properties}

We verify the properties that characterize the canonical gradient (efficient influence function):

\textbf{(A) Mean zero and square integrability.}

We first consider the first term in $\psi(O; \xi, \eta)$. At the truth, for each $(s, v)$, $p_{s}^{(1)}(v) = Z_s(v; p_{s}^{(0)})^T \xi$. Therefore, $p_{s}^{(1)} - Z_s(v; p_{s}^{(0)})^T \xi$ cancels to 0. Thus the first term simplifies to
\begin{equation*}
Z_S(V; p_{S}^{(0)}) \frac{\mathbb{I}\{A=1\}}{\pi_1(S,V)} (Y - p_{S}^{(1)}(V)).
\end{equation*}

Furthermore, conditional on $A = 1$, $S$, and $V$, $\mathbb{E}[Y - p_{S}^{(1)}(V)] = 0$. Thus,
\begin{equation*}
\mathbb{E}\left[Z_S(V; p_{S}^{(0)}) \frac{\mathbb{I}\{A=1\}}{\pi_1(S,V)} (Y - p_{S}^{(1)}(V))\right] = 0.
\end{equation*}

Now we turn to the second term. Conditional on $A = 0$, $S$, and $V$, $\mathbb{E}[Y - p_{S}^{(0)}(V)] = 0$. So the expectation of the second term is 0 as well, implying the desired result that $\mathbb{E}[\psi(O; \xi, \eta)] = 0$.

Since $C$ is finite and invertible and $\psi$ is square-integrable, it follows that
\begin{equation*}
\mathbb{E}[\text{EIF}_\xi] = -C^{-1}\mathbb{E}[\psi(O; \xi, \eta)] = 0,
\end{equation*}
and $\text{EIF}_\xi \in L^2(\mathbb{P}_0)$. Thus $\text{EIF}_\xi$ is a valid mean-zero, square-integrable candidate influence function.

\textbf{(B) Pathwise derivative identity.} Fix any smooth parametric submodel $\mathbb{P}_t$ with $\mathbb{P}_{t=0} = \mathbb{P}_0$ and score function $S(O)$ at $t = 0$. Let $\xi(t) := \xi(\mathbb{P}_t)$ and $\eta(t) := \eta(\mathbb{P}_t)$. Consider the moment map
\begin{equation*}
M(t, \xi) := \mathbb{E}_{\mathbb{P}_t}[\psi(O; \xi, \eta(t))].
\end{equation*}

By identification, for the true parameter $\xi(t)$ we have $M(t, \xi(t)) = 0$ for all $t$. Differentiating at $t = 0$,
\begin{align*}
0 &= \frac{d}{dt} M(t, \xi(t))\Big|_{t=0} \\
&= \frac{\partial}{\partial t} \mathbb{E}_{\mathbb{P}_t}[\psi(O; \xi, \eta(t))]\Big|_{t=0} + \mathbb{E}_{\mathbb{P}_0}\left[\frac{\partial}{\partial \xi} \psi(O; \xi, \eta)\right] \xi'(0).
\end{align*}

For the first term, differentiation under the integral yields
\begin{align*}
\frac{\partial}{\partial t} \mathbb{E}_{\mathbb{P}_t}[\psi(O; \xi, \eta(t))]\Big|_{t=0} &= \mathbb{E}_{\mathbb{P}_0}[\psi(O; \xi, \eta_0) S(O)] \\
&\quad + \frac{d}{dt} \mathbb{E}_{\mathbb{P}_0}[\psi(O; \xi, \eta(t))]\Big|_{t=0}.
\end{align*}

The second term vanishes by Neyman orthogonality (shown below). Hence
\begin{equation*}
\frac{\partial}{\partial t} \mathbb{E}_{\mathbb{P}_t}[\psi(O; \xi, \eta(t))]\Big|_{t=0} = \mathbb{E}_{\mathbb{P}_0}[\psi(O; \xi, \eta) S(O)].
\end{equation*}

Therefore
\begin{equation*}
0 = \mathbb{E}_{\mathbb{P}_0}[\psi(O; \xi, \eta) S(O)] + C \xi'(0),
\end{equation*}
where $C = \mathbb{E}_{\mathbb{P}_0}[\partial_\xi \psi(O; \xi, \eta)]$. Rearranging,
\begin{equation*}
\xi'(0) = -C^{-1}\mathbb{E}_{\mathbb{P}_0}[\psi(O; \xi, \eta) S(O)].
\end{equation*}

But $\text{EIF}_\xi(O) = -C^{-1}\psi(O; \xi, \eta)$. Thus
\begin{equation*}
\mathbb{E}_{\mathbb{P}_0}[\text{EIF}_\xi(O) S(O)] = \xi'(0).
\end{equation*}

\textbf{(C) Efficiency.} The orthogonal score construction ensures that $\text{EIF}_\xi$ is the projection of the pathwise derivative onto the tangent space. The previous identity shows it satisfies the covariance representation for every parametric submodel. Hence it is the canonical gradient and thus the efficient influence function.

\subsubsection*{Neyman Orthogonality}

We have already shown that the estimating equation is unbiased at the truth. It remains to show that
\begin{equation*}
\frac{d}{dt} \mathbb{E}[\psi(O; \xi, \eta + th)]|_{t=0} = 0.
\end{equation*}

Let $q_s(v) = p_{s}^{(0)}(v)$ and $r_s(v) = p_{s}^{(1)}(v)$.

\textbf{Perturbation of $r$.} Write the score as
\begin{align*}
\psi(O; \xi, \eta) &= Z_S(V; q) \left\{ r_S(V) + \frac{\mathbb{I}\{A = 1\}}{\pi_1(S, V)} (Y - r_S(V)) - Z_S(V; q)^T \xi \right\} + \cdots,
\end{align*}
where the second line (omitted here) does not involve $r_S$. The Gateaux derivative in direction $h_r$ contributes
\begin{equation*}
\mathbb{E}\left[Z_S(V; q_0) \left\{ h_r(S, V) - \frac{\mathbb{I}\{A = 1\}}{\pi_{0,1}(S, V)} h_r(S, V) \right\}\right].
\end{equation*}

Since $\mathbb{E}[\mathbb{I}\{A = 1\}/\pi_{0,1}(S, V) | S, V] = 1$, the expression simplifies to zero.

\textbf{Perturbation of $q$.} The score also contains $Z_S(V; q)$ and the residual term $(Y - q_S(V))$. Differentiating these terms with respect to $q$ yields
\begin{align*}
&\mathbb{E}\left[\frac{\partial Z_S(V; q)}{\partial q}\Big|_{q=q_0} \cdot h_q(S, V) \left\{ r_{0,S}(V) - Z_S(V; q_0)^T \xi_0 \right\}\right] \\
&\quad + \mathbb{E}\left[Z_S(V; q_0) [b(V)]^T (\beta_0 - \lambda_0) \frac{\mathbb{I}\{A = 0\}}{\pi_{0,0}(S, V)} (- h_q(S, V))\right].
\end{align*}

The first expectation vanishes because at truth $r_{0,S}(V) = Z_S(V; q_0)^T \xi_0$. The second vanishes because the factor $\mathbb{I}\{A = 0\}/\pi_{0,0}(S, V)$ has conditional expectation 1 given $(S, V)$, while $\mathbb{E}[Y - q_0(S, V) | A = 0, S, V] = 0$.

\textbf{Perturbation of $\pi$.} The score depends on $\pi_a$ only through the inverse weights. Differentiating with respect to $\pi_1$ gives
\begin{equation*}
\mathbb{E}\left[- Z_S(V; q_0) \frac{\mathbb{I}\{A = 1\}}{\pi_{0,1}(S, V)^2} (Y - r_0(S, V)) h_{\pi,1}(S, V)\right].
\end{equation*}

Conditioning on $(S, V, A=1)$, this expectation is zero since $\mathbb{E}[Y - r_0(S, V) | A = 1, S, V] = 0$. The analogous calculation for $\pi_0$ yields zero as well.

\textbf{Conclusion.} Since the Gateaux derivative with respect to each component $r$, $q$, $\pi$ vanishes, we conclude that
\begin{equation*}
\frac{d}{dt} \psi(\xi_0, \eta_t)\Big|_{t=0} = 0.
\end{equation*}

Hence the score $\psi(O; \xi, \eta)$ is Neyman orthogonal.

\subsection*{Generalization of Theorem 2 to Nonlinear Parametric Functions}

\subsubsection*{Setup and Assumptions}

Let $g(v; \beta)$ and $h(v; \lambda)$ be arbitrary smooth parametric functions of a known basis $b(v)$, where $\beta, \lambda \in \mathbb{R}^d$ are finite-dimensional parameter vectors. Examples include:
\begin{itemize}
    \item Linear: $g(v; \beta) = \beta^T b(v)$
\item Logistic: $g(v; \beta) = \text{expit}(\beta^T b(v)) = \frac{1}{1 + \exp(-\beta^T b(v))}$
\item Probit: $g(v; \beta) = \Phi(\beta^T b(v))$
\item Log-linear: $g(v; \beta) = \exp(\beta^T b(v))$
\end{itemize}

Beyond the standard regularity conditions for M-estimation \citep{van2000asymptotic}, we assume that $g(\cdot; \beta)$ and $h(\cdot; \lambda)$ are twice continuously differentiable with respect to $\beta$ and $\lambda$.

Let $\xi = (\beta, \lambda)^T \in \mathbb{R}^{2d}$ denote the full parameter vector as in the main text. Define the generalized design vector:
$$Z_s(v; p_s^{(0)}) = \begin{pmatrix}
(1 - p_s^{(0)}(v)) \cdot \nabla_\beta g(v; \beta) \cdot b(v) \\
p_s^{(0)}(v) \cdot \nabla_\lambda h(v; \lambda) \cdot b(v)
\end{pmatrix}$$

where $\nabla_\beta g(v; \beta)$ and $\nabla_\lambda h(v; \lambda)$ are the derivatives of the link functions with respect to the linear predictors $\beta^T b(v)$ and $\lambda^T b(v)$, respectively.

For common link functions:
\begin{itemize}
\item Linear: $\nabla_\beta g(v; \beta) = 1$
\item Logistic: $\nabla_\beta g(v; \beta) = g(v; \beta)(1 - g(v; \beta))$
\item Probit: $\nabla_\beta g(v; \beta) = \phi(\beta^T b(v))$ where $\phi$ is the standard normal density
\end{itemize}
\begin{theorem}\label{theorem3}
    
\textbf{Theorem 2 (Generalized)} Let $\eta = (p^{(0)}, p^{(1)}, \pi_a(s,v) = P(A = a | S = s, V = v))$ denote the nuisance functions. The efficient influence function (EIF) of $\xi = (\beta, \lambda)$ is:

$$\text{EIF}_\xi(O) = -C^{-1} \psi(O; \xi, \eta)$$

where the estimating function is:

$$\psi(O; \xi, \eta) = Z_S(V; p_S^{(0)}) \left[ p_S^{(1)}(V) + \frac{I\{A = 1\}}{\pi_1(S,V)} (Y - p_S^{(1)}(V)) - Z_S(V; p_S^{(0)})^T \xi \right]$$

$$+ Z_S(V; p_S^{(0)}) [b(V)]^T (\beta - \lambda) \frac{I\{A = 0\}}{\pi_0(S,V)} (Y - p_S^{(0)}(V))$$

and the matrix $C$ is given by:

$$C = -E[Z_S(V; p_S^{(0)}) Z_S(V; p_S^{(0)})^T] + E[\Gamma_S(V; p_S^{(0)})]$$

where $\Gamma_S(V; p_S^{(0)})$ contains second-derivative terms arising from the nonlinearity of the link functions:

$$\Gamma_S(V; p_S^{(0)}) = \begin{pmatrix}
(1 - p_S^{(0)}(V)) \cdot \nabla^2_\beta g(V; \beta) \cdot b(V) b(V)^T & 0 \\
0 & p_S^{(0)}(V) \cdot \nabla^2_\lambda h(V; \lambda) \cdot b(V) b(V)^T
\end{pmatrix}$$

The asymptotic variance of $\hat{\xi}$ satisfying $P_n \psi(O; \hat{\xi}, \hat{\eta}) = 0$ is $C^{-1} \text{Var}(\psi(O; \xi, \eta)) C^{-1}$, and the estimating equation is Neyman-orthogonal.
\end{theorem}

\begin{proof}

Part I: Efficient Influence Function Properties

We verify the three properties that characterize the canonical gradient:

(A) Mean zero

The first term in $\psi(O; \xi, \eta)$ simplifies at the truth. Since $p_S^{(1)}(v) = g(v; \beta)(1 - p_S^{(0)}(v)) + h(v; \lambda) p_S^{(0)}(v)$, we have:
$$p_S^{(1)}(V) - Z_S(V; p_S^{(0)})^T \xi = 0$$

Thus the first term becomes:
$$Z_S(V; p_S^{(0)}) \frac{I\{A = 1\}}{\pi_1(S,V)} (Y - p_S^{(1)}(V))$$

Taking expectations and using $E[Y - p_S^{(1)}(V) | A = 1, S, V] = 0$:
$$E\left[Z_S(V; p_S^{(0)}) \frac{I\{A = 1\}}{\pi_1(S,V)} (Y - p_S^{(1)}(V))\right] = 0$$

Similarly, the second term has expectation zero since $E[Y - p_S^{(0)}(V) | A = 0, S, V] = 0$.

(B) Pathwise derivative identity

Consider any smooth parametric submodel $P_t$ with $P_{t=0} = P_0$ and score function $S(O)$ at $t = 0$. Let $\xi(t) = \xi(P_t)$ and $\eta(t) = \eta(P_t)$. 

The moment condition $M(t, \xi) = E_{P_t}[\psi(O; \xi, \eta(t))] = 0$ holds for the true parameter path $\xi(t)$.

Differentiating at $t = 0$:
$$0 = \frac{d}{dt} M(t, \xi(t))\Big|_{t=0} = \frac{\partial}{\partial t} E_{P_t}[\psi(O; \xi, \eta(t))]\Big|_{t=0} + E_{P_0}\left[\frac{\partial}{\partial \xi} \psi(O; \xi, \eta)\right] \xi'(0)$$

The first term equals $E_{P_0}[\psi(O; \xi, \eta) S(O)]$ by differentiation under the integral plus the Neyman-orthogonality property (shown below).

With $C = E_{P_0}[\partial_\xi \psi(O; \xi, \eta)]$, we obtain:
$$\xi'(0) = -C^{-1} E_{P_0}[\psi(O; \xi, \eta) S(O)]$$

This establishes the pathwise derivative representation.

(C) Efficiency

The orthogonal score construction ensures that $\text{EIF}_\xi$ achieves the semiparametric efficiency bound by construction.

\textbf{Neyman Orthogonality}

We must show $\frac{d}{dt} E[\psi(O; \xi, \eta + th)]|_{t=0} = 0$ for any direction $h$.

Perturbation of $p^{(1)}$ (nuisance function $r$):

The Gâteaux derivative in direction $h_r$ contributes:
$$E\left[Z_S(V; p_S^{(0)}) \left(h_r(S,V) - \frac{I\{A = 1\}}{\pi_1(S,V)} h_r(S,V)\right)\right]$$

Since $E[I\{A = 1\}/\pi_1(S,V) | S,V] = 1$, this expression equals zero.

Perturbation of $p^{(0)}$ (nuisance function $q$):

This is the most complex part due to the nonlinear link functions. The perturbation affects both $Z_S(V; q)$ and the residual terms.

The derivative of $Z_S(V; q)$ with respect to $q$ involves the second derivatives of the link functions:
$$\frac{\partial Z_S(V; q)}{\partial q} = \begin{pmatrix}
-\nabla_\beta g(V; \beta) \cdot b(V) + (1-q_S(V)) \cdot \nabla^2_\beta g(V; \beta) \cdot b(V) \\
\nabla_\lambda h(V; \lambda) \cdot b(V) + q_S(V) \cdot \nabla^2_\lambda h(V; \lambda) \cdot b(V)
\end{pmatrix}$$

The contribution from perturbing $q$ is:
$$E\left[\frac{\partial Z_S(V; q)}{\partial q}\Big|_{q=q_0} \cdot h_q(S,V) \cdot (p_{0,S}^{(1)}(V) - Z_S(V; q_0)^T \xi_0)\right]$$

At the truth, $p_{0,S}^{(1)}(V) = Z_S(V; q_0)^T \xi_0$, so this term vanishes.

The remaining orthogonality correction term:
$$E\left[Z_S(V; q_0)[b(V)]^T (\beta_0 - \lambda_0) \frac{I\{A = 0\}}{\pi_0(S,V)} (-h_q(S,V))\right]$$

vanishes because $E[Y - q_0(S,V) | A = 0, S, V] = 0$.

Perturbation of propensity scores $\pi_a$:

The score depends on $\pi_a$ only through inverse weights. Differentiating yields terms of the form:
$$E\left[-Z_S(V; q_0) \frac{I\{A = a\}}{\pi_a(S,V)^2} (\text{residual}) \cdot h_{\pi,a}(S,V)\right]$$

These vanish because the residuals have conditional expectation zero given $(A = a, S, V)$.

\textbf{Matrix $C$ Calculation}

The matrix $C = E[\partial_\xi \psi(O; \xi, \eta)]$ contains two components:

1. Main term: $-E[Z_S(V; p_S^{(0)}) Z_S(V; p_S^{(0)})^T]$

2. Second-derivative correction: The nonlinearity of link functions contributes additional terms through:
   $$E\left[\frac{\partial Z_S(V; p_S^{(0)})}{\partial \xi} \cdot (p_S^{(1)}(V) - Z_S(V; p_S^{(0)})^T \xi)\right]$$
   
   At the truth, the second factor vanishes, but the derivative itself contributes $E[\Gamma_S(V; p_S^{(0)})]$ where $\Gamma$ contains the second derivatives of the link functions.

This completes the proof. 
\end{proof}

\textbf{Corollary: Specific Link Functions}
\begin{itemize}
    \item Logistic Links: When $g(v; \beta) = \text{expit}(\beta^T b(v))$ and $h(v; \lambda) = \text{expit}(\lambda^T b(v))$:
- $\nabla_\beta g(v; \beta) = g(v; \beta)(1 - g(v; \beta))$
- $\nabla^2_\beta g(v; \beta) = g(v; \beta)(1 - g(v; \beta))(1 - 2g(v; \beta))$
\item Linear Links: When $g(v; \beta) = \beta^T b(v)$ and $h(v; \lambda) = \lambda^T b(v)$:
- $\nabla_\beta g(v; \beta) = 1$ and $\nabla^2_\beta g(v; \beta) = 0$
- $\Gamma_S(V; p_S^{(0)}) = 0$, recovering the original Theorem 2
\end{itemize}

As mentioned in the main text, since the design vector $Z_S(V; p_S^{(0)})$ depends on the unknown parameters $\xi$, estimation requires an iterative procedure (e.g., Newton-Raphson) that alternates between updating the design vector and solving the estimating equation until convergence.
\end{document}